\documentclass[10pt,reprint, pra, aps]{revtex4-1}
\usepackage[utf8]{inputenc}
\usepackage[english]{babel}
\usepackage{amsmath}
\usepackage{amsfonts}
\usepackage{amssymb}
\usepackage{makeidx}
\usepackage{csquotes}
\usepackage{mathrsfs}
\usepackage{amsthm}
\usepackage{color}
\usepackage{graphicx}
\usepackage{epstopdf}
\usepackage{enumerate}
\newtheorem{prop}{Proposition}

\theoremstyle{remark}
\newtheorem{exm}{Example}
\newtheorem{rem}{Remark}
\theoremstyle{definition}
\newtheorem{defin}{Definition}

\def\intr{\mathit{int}}

\def\Pe{\mathcal{P}}
\def\Ka{K}
\def\Be{\mathcal{B}}
\def\<{\langle}
\def\>{\rangle}

\usepackage{dsfont}

\DeclareMathOperator{\degcom}{DegCom}
\DeclareMathOperator{\conv}{\mathit{conv}}
\DeclareMathOperator{\aff}{\mathit{aff}}
\begin{document}

\title{All measurements in a probabilistic theory are compatible if and only if the state space is a simplex}

\author{Martin Pl\'avala}
\affiliation{Mathematical Institute, Slovak Academy of Sciences, \v Stef\' anikova 49, Bratislava, Slovakia}

\begin{abstract}
We study the compatibility of measurements on finite-dimensional compact convex state space in the framework of general probabilistic theory. Our main emphasis is on formulation of necessary and sufficient conditions for two-outcome measurements to be compatible and we use these conditions to show that there exist incompatible measurements whenever the state space is not a simplex. We also formulate the linear programming problem for the compatibility of two-outcome measurements.
\end{abstract}

\maketitle

\section{Introduction}
Incompatibility lies deeply within quantum mechanics and many of the famous and key aspect of quantum theories have been traced to Heisenberg uncertainty principle, no cloning theorem, violations of Bell inequalities and other notions making use of compatibility, see \cite{HeinosaariMiyaderaZiman-compatibility} for recent review. In light of these discoveries compatibility in the framework of general probabilistic theories has been studied \cite{WolfPerezgarciaFernandez-measIncomp, BarnumHowardBarretLeifer-noBroadcast, BuschHeinosaariSchultzStevens-compatibility} in order to show the difference between classical and non-classical probabilistic theories. Also the connection of compatibility and steering in general probabilistic theories have been studied \cite{Banik-steering, UolaMoroder-steering}.

Recently incompatibility of measurements on quantum channels and combs has been in question \cite{SedlakReitznerChiribellaZiman-compatibility} as it potentially could be used as a resource in quantum theory in a similar ways as an incompatibility of measurements on quantum states \cite{HeinosaariMiyaderaZiman-compatibility}. The degree of compatibility (also called robustness of incompatibility) has been studied for measurements on channels and combs \cite{HeinosaariKiukasReitzner-robustness, WolfPerezgarciaFernandez-measIncomp, SedlakReitznerChiribellaZiman-compatibility}.

In the present article we study the notion of compatibility of measurements in the framework of probabilistic theories and we show that every two measurements are compatible if and only if the state space is a simplex. In one way this result has clear physical interpretation - classical state space is always a simplex and the existence of incompatible measurements is often seen as one of the main aspects of quantum theories.

The paper is organized as follows: Sec. \ref{sec:prelim} contains preliminary mathematical results and references. Note that Subsec. \ref{subsec:prelim-faces} contains the definition of maximal face that (to the best knowledge of present author) was not defined elsewhere (even though it has close tie to the notion of tangent half-space and tangent hyper-plane \cite[pp. 169]{Rockafellar-convex}) and is later used in Sec. \ref{sec:inc}. In Sec. \ref{sec:meas} the measurements are defined. In Sec. \ref{sec:inc} compatibility of measurements and degree of compatibility is defined and it is shown that all measurements are compatible if and only if the state space is a simplex. Also the linear program for compatibility of two two-outcome measurements is formulated.

\section{Preliminaries} \label{sec:prelim}
We present preliminary mathematical knowledge used in the paper. In all of the paper $E$ will denote a real, finite dimensional vector space equipped with the Euclidean topology and $\Ka$ will denote non-empty compact convex subset of $E$. We will denote the convex hull of a set $X$ as $\conv(X)$, affine hull of a set $X$ as $\aff(X)$, interior of a set of a set $X$ as $\intr(X)$ and by $\partial \Ka$ we will denote the boundary of $\Ka$, i.e $\partial \Ka = \Ka \setminus \intr(\Ka)$ as $\Ka$ is closed.

\subsection{Structure of $A(\Ka)$} \label{subsec:structure}
By $A(\Ka)$ we will denote the set of real valued affine functions on $\Ka$ and by $A(\Ka)^+$ we will denote the set of positive affine functions on $\Ka$, i.e. $f \in A(\Ka)^+$ if and only if $f(x) \geq 0$ for every $x \in \Ka$. We will denote constant functions by the value they attain. Since $\Ka$ is compact and the functions $A(\Ka)$ are continuous, every function reaches its maximum and minimum over $\Ka$ at some point of $\Ka$ and we can introduce the supremum norm for $f \in A(\Ka)$ as
\begin{equation*}
\Vert f \Vert_A = \sup_{x \in \Ka} |f(x)|.
\end{equation*}

The set $A(\Ka)^+$ is:
\begin{itemize}
\item closed
\item convex, i.e. for $\lambda \in \mathbb{R}$, $0 \leq \lambda \leq 1$,$f_1, f_2 \in A(\Ka)^+$ we have $\lambda f_1 + (1-\lambda)f_2 \in A(\Ka)^+$
\item cone, i.e. for $\nu \in \mathbb{R}$, $f \in A(\Ka)^+$ we have $\nu f_1 \in A(\Ka)^+$
\item pointed, i.e. $A(\Ka)^+ \cap (-A(\Ka)^+) = \{0\}$
\item generating, i.e. for every $f \in A(\Ka)$ we have $f_+, f_- \in A(\Ka)^+$ such that $f = f_+ - f_-$.
\end{itemize}
The closed, pointed, convex cone $A(K)^+$ defines a partial order $\geq$ on $A(\Ka)$ given for $f_1, f_2 \in A(\Ka)$ as
\begin{equation*}
f_1 \geq f_2 \Leftrightarrow f_1 - f_2 \in A(\Ka)^+
\end{equation*}
or equivalently $f_1 \geq f_2 \Leftrightarrow (f_1 - f_2)(x) \geq 0, \forall x \in \Ka$. The partial order $\geq$ will play a role in our formulation of linear program for incompatibility of two-outcome measurements.

\begin{defin}
We say that $e \in A(\Ka)^+$ is an order unit if for every $f \in A(\Ka)^+$ there is some $\nu \in \mathbb{R}$, $\nu > 0$ such that
\begin{equation*}
\nu e \geq f.
\end{equation*}
\end{defin}
In the current setting it is easy to see that every strictly positive function is an order unit. We will omit the simple proof of the following fact.
\begin{prop}
$e \in A(\Ka)^+$ is an order unit if and only if $e \in \intr(A(\Ka)^+)$.
\end{prop}

We will also use the notion of a base of a cone.
\begin{defin}
Let $Q \subset E$ be a cone, then a set $\Be \subset Q$ is called base of $Q$ if for every $0 \neq x \in Q$ there exist unique $y \in \Be$ and $\lambda \in \mathbb{R}$ such that $x = \lambda y$.
\end{defin}

To formulate the linear programming problem we will also have to work with the dual space of $A(\Ka)$, we will denote it $A(\Ka)^*$. We will denote by $A(\Ka)^{*+}$ the cone of positive functionals dual to $A(\Ka)^+$, that is $\psi \in A(\Ka)^{*+}$ if and only if for every $f \in A(\Ka)^+$ we have $\psi(f) \geq 0$.
\begin{prop}
$A(\Ka)^{*+}$ is a closed pointed convex cone.
\end{prop}
\begin{proof}
It is straightforward to see that $A(\Ka)^{*+}$ is closed convex cone. It is pointed because $A(K)^+$ is generating.
\end{proof}

We define the dual norm for $\psi \in A(\Ka)^*$ as
\begin{equation*}
\Vert \psi \Vert_* = \sup_{\Vert f \Vert_A \leq 1} |\psi(f)|.
\end{equation*}

For $x \in \Ka$ let $\Phi_x \in A(\Ka)^*$ be given for $f \in A(\Ka)$ as
\begin{equation*}
\Phi_x (f) = f(x).
\end{equation*}
The map $\Phi: \Ka \to A(\Ka)^*$ is called evaluation map and it is affine. It is easy to see that $\Phi[\Ka] = \{ \Phi_x : x \in \Ka \}$ contains only positive functionals with unit norm such that $\Phi_x(1) = 1$ for every $x \in \Ka$. The converse is also true:
\begin{prop}
$\Phi[\Ka] = \{ \psi \in A(B)^* : \Vert \psi \Vert_* = \psi(1) = 1 \}$.
\end{prop}
\begin{proof}
For proof see \cite[Theorem 4.3]{AsimowEllis}. Also note that $\Vert \psi \Vert_* = \psi(1) = 1$ implies $\psi \geq 0$.
\end{proof}
The set $\Phi[\Ka]$ is sometimes referred to as the state space as in general applications it is often easier to work with $\Phi[\Ka]$ rather than $\Ka$.

\subsection{Exposed faces and maximal faces of a convex set} \label{subsec:prelim-faces}
In this subsection we will define faces, exposed faces and maximal faces and prove Prop. \ref{prop:introd-faces-pointInMaximal}.
\begin{defin}
Let $C \subset \Ka$ be a convex set (that is $C$ is a convex set that is subset of $\Ka$). We say that $C$ is a face of $\Ka$ if $x \in C$, $\lambda \in \mathbb{R}$, $0 < \lambda < 1$ and $x = \lambda y + (1-\lambda) z$ implies $y, z \in C$.
\end{defin}
It is straightforward that $\Ka$ and the empty set are a faces of $\Ka$ and they are called the trivial faces. Apart from the trivial faces it is known that all faces lie in $\partial \Ka$ \cite[Corollary 18.1.3]{Rockafellar-convex}. Face consisting of only single point is called extreme point of $\Ka$.
\begin{defin}
Let $C \subset K$ be a set where some affine function $f$ reaches its maximum (or minimum) over $\Ka$, i.e. if $max_{x \in \Ka} f(x) = M_f$, then $C = \{ x \in \Ka : f(x) = M_f \}$. Such $C$ is called exposed face of $\Ka$.
\end{defin}
Every exposed face is a face \cite[pp. 162]{Rockafellar-convex}. An exposed face consisting of only single point will be called exposed point. It will be important that the set of exposed points of $\Ka$ is dense in the set of extreme points of $\Ka$ \cite[Theorem 18.6]{Rockafellar-convex} and that every face of a closed convex set is closed \cite[Corollary 18.1.1]{Rockafellar-convex}. Also note that not every extreme point must be an exposed point, example of this is presented in \cite[pp. 163]{Rockafellar-convex}.

We proceed by defining the notion of maximal face. Maximal faces are generalization of the $n-1$ dimensional exposed faces of polytopes (that is of convex sets that are convex hull of finite number of points).
\begin{defin} \label{def:introd-faces-maximal}
Let $C \subset \Ka$ be a nontrivial face, such that for every $x \in \Ka \setminus C$ we have $\conv(C \cup \{x\}) \cap \intr(\Ka) \neq \emptyset$, then we say that $C$ is a maximal face.
\end{defin}
Note that we require maximal faces to be nontrivial, i.e. $\Ka$ itself is not a maximal face. One can show that every maximal face is exposed, because every maximal face is an intersection of $\Ka$ and a hyper-plane tangent to $\Ka$. Also every intersection of $\Ka$ and hyper-plane tangent to $\Ka$ is a maximal faces. We present a simple example of maximal faces of triangle and circle.
\begin{exm}
Assume that $\Ka \subset \mathbb{R}^2$ is a triangle. The vertices of the triangle are extreme and exposed points of $\Ka$, but they are not maximal faces. In this case maximal faces are the edges of the triangle.

Now consider that $\Ka \subset \mathbb{R}^2$ is the convex hull of the unit circle, then every extreme point of $\Ka$ is a maximal face.
\end{exm}

Maximal faces will play a role in the notion of compatibility of measurements as the condition $\conv(C \cup \{x\}) \cap \intr(\Ka) \neq \emptyset$ will be of great importance. 
\begin{prop} \label{prop:introd-faces-pointInMaximal}
Let $\Ka \subset \mathbb{R}^n$ be a non-empty convex compact set. Then for every point $x \in \partial \Ka$ there are maximal faces $C_1, C_2$ such that $x \in C_1$ and $x \notin C_2$.
\end{prop}
\begin{proof}
We will prove the statement in two steps. As first we will prove that  that every point of $\partial \Ka$ belongs to some maximal face. Then we prove that maximal faces that have a point in common can not form $\partial \Ka$.

Let $x \in \partial \Ka$, then there exists a non-constant affine function $f$ that reaches its maximum over $\Ka$ in $x$ \cite[Colloraly 11.6.2]{Rockafellar-convex}, let $f(x) = M_f$. The set $G_0 = \{ x' \in \Ka : f(x') = M_f \}$ is an exposed face. If $G_0$ is maximal face then we are done, if $G_0$ is not a maximal face, then there must exist a point $y \in \Ka \setminus G_0$ such that $\conv(G_0 \cup \{y\}) \cap \intr(\Ka) = \emptyset$. The set $\conv(G_0 \cup \{y\})$ does not have to be face itself, but since $\conv(G_0 \cup \{y\}) \cap \intr(\Ka) = \emptyset$ then there exists a non-trivial supporting hyperplane to $\Ka$ containing $\conv(G_0 \cup \{y\})$, see \cite[Theorem 11.6]{Rockafellar-convex} for definition of supporting hyperplane to $\Ka$ and proof of the statement. In other words there must exist a non-constant affine function $f_1$ such that $\max_{y \in \Ka} f_1(y) = M_{f_1}$ and $G_1 = \{ x' \in \Ka : f_1(x') =  M_{f_1}\} \supset \conv(G_0 \cup \{y\})$, i.e. $G_1$ is an exposed face of $\Ka$ and $x \in G_1$. Moreover for the dimensions of $\aff(G_0)$ and $\aff(G_1)$ we must have $\dim(\aff(G_1)) > \dim(\aff(G_0))$, because $y \in G_1$ and $y \notin G_0$. If $G_1$ is a maximal face then we are finished, if not then we can repeat the procedure to find exposed face $G_2 \supset G_1$.

Since the affine span of every maximal face can be at most $n-1$ dimensional and the dimension of affine span of the exposed faces $G_i$ is strictly growing with $i$ it is clear that we can repeat this procedure at most $n-1$ times to obtain a maximal face, hence in this way to every $x \in \partial \Ka$ we can find a maximal face that contains it.

Now we will proceed with the second part of the proof. Take $x \in \partial \Ka$, denote $\{ C_i \}$ the set of all maximal faces of $\Ka$ and assume $x \in \cap_{i} C_i$. Since every point of $\partial \Ka$ belongs to some maximal face we must have $\cup_i C_i = \partial \Ka$. Let us define positive affine functions $f_i$, such that $C_i = \{ y \in \Ka : f_i (y) = 0 \}$ then since a finite dimensional convex compact set is an intersection of closed half-spaces tangent to it \cite[Theorem 18.8]{Rockafellar-convex} we have $\Ka = \{ y \in \mathbb{R}^n : f_i (y) \geq 0, \; \forall i \}$. Since we have $f_i(x) = 0$, $\forall i$ then for any $\lambda \in \mathbb{R}$, $\lambda \geq 0$ and $z \in \Ka$ we have
\begin{equation*}
f_i ( \lambda z + (1-\lambda) x) = \lambda f_i(z) \geq 0
\end{equation*}
for every $i$. This implies that $\lambda z + (1-\lambda) x \in \Ka$ which is in contradiction with $\Ka$ being compact.
\end{proof}

\section{Measurements on $\Ka$} \label{sec:meas}
Let $E$ be a finite dimensional real vector space equipped with the Euclidean topology and let $\Ka \subset E$ be a compact convex set. We will call $\Ka$ a state space as it represents a set of all possible states of some system and the convex combination is interpreted as probabilistic mixture. Let $\Omega$ be a nonempty compact Hausdorff space and let $\mathcal{P}(\Omega)$ denote the set of Borel probability measures on $\Omega$.
\begin{defin}
Measurements (also called observables) on $\Ka$ with sample space $\Omega$ are affine mappings $m: \Ka \to \Pe(\Omega)$.
\end{defin}
The interpretation is that $\Omega$ represents all possible outcomes of a certain measurement and is usually referred to as sample space. For $x \in \Ka$ the measure $m(x) \in \Pe (\Omega)$ is a generalized notion of assigning probabilities to the measurement outcomes. Our definition follows the usual definitions of measurements in probabilistic theories \cite{BuschHeinosaariSchultzStevens-compatibility, Banik-steering} but may be easily generalized to locally compact sample spaces $\Omega$. Let $\sigma \subset \Omega$ be a measurable set, then by $m(x; \sigma)$ we will denote the measure of the set $\sigma$ with respect to the measure $m(x)$.

\subsection{Finite outcome measurements}
Let the sample space $\Omega = \{ \omega_1, \ldots, \omega_k \}$ be a finite set. Every Borel probability measure $\mu \in \Pe (\Omega)$ is of the form
\begin{equation*}
\mu = \sum_{i=1}^k \lambda_i \delta_{\omega_i}
\end{equation*}
where $\delta_{\omega_i}$ is the Dirac measure centered at $\omega_i$ and $\lambda_i \in \mathbb{R}$, $0 \leq \lambda_i \leq 1$, $\sum_{i=1}^k \lambda_i = 1$. It follows that if $m$ is a measurement on $\Ka$ with finite sample space $\Omega$ then there always are functions $f_j \in A(\Ka)^+$, $0 \leq f_j \leq 1$ for $j \in \{1, \ldots, k\}$, $\sum_{j=1}^k f_j = 1$ such that
\begin{equation*}
m = \sum_{i=1}^k f_i \delta_{\omega_i}.
\end{equation*}

\begin{rem}
In the standard literature \cite{HeinosaariZiman-MLQT, Holevo-QT} usually it is instead of writing $m = \sum_{i=1}^k f_i \delta_{\omega_i}$ simply said that the function $f_j$ represents the probability of the outcome $\omega_j$. To simplify the notation we will use the formulation presented above.
\end{rem}

\section{Compatibility of measurements} \label{sec:inc}
Assume that we wish to perform two distinct measurements $m_1, m_2$ with two separate sample spaces $\Omega_1, \Omega_2$. We would like to know whether there exists a measurement that performs both $m_1$ and $m_2$ at the same time. To ask this question properly we will introduce the concept of marginal measurement. When working with the Cartesian product $\Omega_1 \times \Omega_2$ we will always consider the product topology on it given by the topologies of $\Omega_1, \Omega_2$.

\begin{defin}
Let $m : \Ka \to \Pe (\Omega_1 \times \Omega_2)$ be a measurement on $\Ka$ with sample space $\Omega_1 \times \Omega_2$. We say that $m_1 : \Ka \to \Pe(\Omega_1)$ is a marginal measurement of $m$ if for every measurable set $\sigma \subset \Omega_1$ and $x \in \Ka$ we have
\begin{equation*}
m_1(x; \sigma) = m(x; \sigma \times \Omega_2).
\end{equation*}
\end{defin}
This definition can be formally understood as
\begin{equation*}
m_1(x; \sigma) = \int_{\Omega_2} m(x; \sigma \times d \omega_2)
\end{equation*}
for every measurable set $\sigma \subset \Omega_1$. For the finite outcome measurements the integral is replaced by a sum over the outcomes, i.e. for $m = \sum_{i,j=1}^k f_{ij} \delta_{(\omega_i, \omega_j)}$, where $f_{ij} \in A(K)^+$ and $\delta_{(\omega_i, \omega_j)}$ is the Dirac measure centered at $(\omega_i, \omega_j)$, we have
\begin{equation*}
m_1 = \sum_{i,j=1}^k f_{ij} \delta_{\omega_i}.
\end{equation*}

It is straightforward to see that $m_1$ is a measurement on $\Ka$ with sample space $\Omega_1$ as the positivity and normalization to 1 follow from the properties of $m$. Now we are ready for the definition of compatibility.

\begin{defin}
We will say that measurements $m_1 : \Ka \to \Pe(\Omega_1)$ and $m_2 : \Ka \to \Pe(\Omega_2)$ are compatible if there exists a measurement $m: \Ka \to \Pe(\Omega_1 \times \Omega_2)$ such that $m_1, m_2$ are marginal measurements of $m$.
\end{defin}

This definition is the standard definition used for compatibility of measurements.

A natural question is: are there any incompatible measurements? It is of course long known that incompatible measurements in quantum mechanics exist, but mathematically it is interesting to ask what properties of $\Ka$ imply that all measurements are compatible.
\begin{prop} \label{prop:inc-simplex}
Let $\Ka$ be a simplex, that is let $\{x_1, \ldots, x_n\}$ be the set of extreme points of $\Ka$ such that the points $x_1, \ldots, x_n$ are affinely independent. Then every measurement on $\Ka$ is compatible with every other measurement on $\Ka$.
\end{prop}
\begin{proof}
Let $\Ka$ be a simplex then there exists affine functions $b_j: \Ka \to \mathbb{R}$, $j \in \{1, \ldots, n\}$ defined by $b_j(x_i) = \delta_{ij}$. These functions are positive, because for every $y \in \Ka$ we have $y = \sum_{i=1}^n \lambda_i x_i$ with $\sum_{i=1}^n \lambda_i = 1$ and $0 \leq \lambda_i \leq 1$ for every $i$.

Let $m$ be a measurement on $\Ka$ with a sample space $\Omega$, then for $y \in \Ka$, $y = \sum_{i=1}^n \lambda_i x_i$ we have
\begin{equation*}
m(y) = \sum_{i=1}^n \lambda_i m(x_i) = \sum_{i=1}^n b_i(y) m(x_i),
\end{equation*}
i.e. a measurement $m$ on simplex is uniquely described by the measures $m(x_i) \in \Pe(\Omega)$.

Now let $m_1, m_2$ be measurements on $\Ka$ with the sample spaces $\Omega_1, \Omega_2$ respectively, then for $y \in \Ka$ we have as above
\begin{equation*}
m_j(y) = \sum_{i=1}^n b_i(y) m_j(x_i),
\end{equation*}
for $j \in \{1, 2\}$. Let $(m_1 \times m_2)(x_i)$ denote the product measure obtained form the measures $m_1(x_i)$ and $m_2(x_i)$, that is for measurable sets $\sigma_i \subset \Omega_i$, $i \in \{1, 2\}$ we have
\begin{equation*}
(m_1 \times m_2)(x_i; \sigma_1 \times \sigma_2) = m_1(x_i; \sigma_1) m_1(x_i; \sigma_2).
\end{equation*}
Let the measurement $m: \Ka \to \Pe(\Omega_1 \times \Omega_2)$ be given as
\begin{equation*}
m(y) = \sum_{i=1}^n b_i(y) (m_1 \times m_2)(x_i)
\end{equation*}
then it is easy to verify that $m_1$ and $m_2$ are marginal measurements of $m$.
\end{proof}
Note that positivity of functions $b_j$ plays a crucial role in the proof and these functions are positive only if $\Ka$ is a simplex. Next we introduce the concept of a coin-toss (also called trivial) measurement.
\begin{defin}
Let $\mu$ be some fixed Borel probability measure on sample space $\Omega$, then by coin-toss we will refer to the measurement given as
\begin{equation*}
m(y) = \mu
\end{equation*}
for every $y \in \Ka$.
\end{defin}
Coin-toss measurements usually represent noise, that is some random factor that affects the measurement outcomes. It can be also interpreted as the most simple measurement when we ignore any information about the state and simply "toss a coin" and return whatever value we obtain. It is straightforward that any coin-toss measurement is compatible with any other measurement.

In the following we state the usual definition of the degree of compatibility.
\begin{defin}
Let $i \in \{1, 2\}$ and let $m_i :\Ka \to \Pe(\Omega_i)$ be a measurement on $\Ka$ with sample space $\Omega_i$. Let $\tau_i : \Ka \to \Pe(\Omega_i)$ be some coin-toss measurements, then we define degree of compatibility of measurements $m_1, m_2$ as
\begin{align*}
\degcom (m_1, m_2) =& \sup_{\substack{0 \leq \lambda \leq 1 \\ \tau_1, \tau_2}} \{ \lambda : \lambda m_1 + (1-\lambda) \tau_1, \\ & \lambda m_2 + (1-\lambda) \tau_2 \; \text{are compatible} \}.
\end{align*}
\end{defin}
The reason for considering different trivial measurements $\tau_1, \tau_2$ is that the sample spaces may be different and even if they would be the same due to our definitions we can not pick some preferred measure as for example properly normed Lebesgue measure on a compact subset of $\mathbb{R}^k$. Note that the supremum is taken also over the coin-toss measurements $\tau_1, \tau_2$.

Based on the analysis of compatibility presented in \cite{HeinosaariMiyaderaZiman-compatibility} we obtain the following:
\begin{prop} \label{prop:inc-minimal-deg}
For any two measurements $m_i :\Ka \to \Pe(\Omega_i)$, $i \in \{1, 2\}$, we have $\degcom(m_1, m_2) \geq \frac{1}{2}$.
\end{prop}
\begin{proof}
The idea is that we can always toss a fair two sided coin, based on the result implement one of the measurements and substitute the other by the respective coin-toss observable. In other words let $\mu_1, \mu_2$ be any Borel probability measures on $\Omega_1, \Omega_2$ respectively that give rise to coin-toss measurements $\tau_i$ given as $\tau_i(y) = \mu_i$, $i \in \{1, 2\}$. Consider the measurement $m: \Ka \to \Pe(\Omega_1 \times \Omega_2)$ given for $y \in \Ka$ as
\begin{equation*}
m(y) = \dfrac{1}{2} \left( \mu_1 \times m_2(y) + m_1(y) \times \mu_2 \right).
\end{equation*}
It is straightforward to verify that the measurements $\frac{1}{2}(m_1 + \tau_1)$ and $\frac{1}{2}(m_2 + \tau_2)$ are marginal measurements of $m$.
\end{proof}
Similar result has been observed even for compatibility of quantum channels \cite{HeinosaariMiyadera-compOfChan}.

\subsection{Compatibility of two-outcome measurements} \label{subsec:inc-two-outcome}
In general it may be hard to decide whether measurements $m_1$ and $m_2$ are compatible but in the case of two-outcome measurements, that is in the case when $\Omega_1, \Omega_2$ contain only two points, we will formulate necessary and sufficient conditions for the measurements $m_1, m_2$ to be compatible. These conditions may be generalized in the same manner to general finite outcome measurements.

Let $\Omega_1 = \Omega_2 = \Omega = \{ \omega_1, \omega_2\}$ be the sample space of the measurements $m_1, m_2$, then they are of the form
\begin{align*}
m_i &= f_i \delta_{\omega_1} + (1-f_i)\delta_{\omega_2}
\end{align*}
for $i \in \{1, 2\}$. Also every measurement $m$ on $\Ka$ with sample space $\Omega \times \Omega$ is of the form
\begin{equation*}
m = g_{11} \delta_{(\omega_1, \omega_1)} + g_{12} \delta_{(\omega_1, \omega_2)} + g_{21} \delta_{(\omega_2, \omega_1)} + g_{22} \delta_{(\omega_2, \omega_2)},
\end{equation*}
where $g_{11}, g_{12}, g_{21}, g_{22} \in A(\Ka)^+$ and $\delta_{(\omega_j, \omega_k)}$ is a Dirac measure on $\Omega \times \Omega$ centered at $(\omega_j, \omega_k) \in \Omega \times \Omega$. Assume that $m_1$ and $m_2$ are marginal measurements of $m$, then we obtain
\begin{align}
g_{11} + g_{12} &= f_1, \label{eq:inc-two-outcome-sigmas-1} \\
g_{21} + g_{22} &= 1-f_1, \label{eq:inc-two-outcome-sigmas-2} \\
g_{11} + g_{21} &= f_2, \label{eq:inc-two-outcome-sigmas-3} \\
g_{12} + g_{22} &= 1-f_2. \label{eq:inc-two-outcome-sigmas-4}
\end{align}
These equations imply $g_{11} + g_{12} + g_{21} + g_{22} = 1$, but not $g_{jk} \geq 0$, $j,k \in \{1,2\}$ and they in general don't have a unique solution. Let $g_{11} = p$, $0 \leq p \leq 1$, then a general solution to Eq. \eqref{eq:inc-two-outcome-sigmas-1} - \eqref{eq:inc-two-outcome-sigmas-4} is
\begin{align*}
g_{12} &= f_1 - p, \\
g_{21} &= f_2 - p, \\
g_{22} &= 1 - f_1 - f_2 + p,
\end{align*}
which imply the inequalities
\begin{align}
f_1 &\geq p, \label{eq:inc-two-outcome-ineq-1} \\
f_2 &\geq p, \label{eq:inc-two-outcome-ineq-2} \\
1 + p &\geq f_1 + f_2, \label{eq:inc-two-outcome-ineq-3}
\end{align}
that come from $g_{jk} \geq 0$ for all $j,k \in \{1,2\}$. In general there may not exist such $p$ satisfying Ineq. \eqref{eq:inc-two-outcome-ineq-1} - \eqref{eq:inc-two-outcome-ineq-3}, in that case the measurements are incompatible. But if $m$ is a joint measurement of $m_1, m_2$ then the Ineq. \eqref{eq:inc-two-outcome-ineq-1} - \eqref{eq:inc-two-outcome-ineq-3} must be satisfied and Eq. \eqref{eq:inc-two-outcome-sigmas-1} - \eqref{eq:inc-two-outcome-sigmas-4} are satisfied simply because $m_1$ and $m_2$ are marginals of $m$. We have proved the following:
\begin{prop} \label{prop:inc-two-outcome-iff}
Let $m_1, m_2$ be two-outcome measurements on $\Ka$ given as
\begin{align*}
m_i &= f_i \delta_{\omega_1} + (1-f_i)\delta_{\omega_2}
\end{align*}
for $i \in \{1, 2\}$, then they are compatible if and only if there is a  function $p \in A(\Ka)^+$, such that $0 \leq p \leq 1$ and Ineq. \eqref{eq:inc-two-outcome-ineq-1} - \eqref{eq:inc-two-outcome-ineq-3} are satisfied.
\end{prop}
Similar results in terms of operators in case of measurements on states were obtained in \cite{HeinosaariKiukasReitzner-robustness, WolfPerezgarciaFernandez-measIncomp}.

Now we will proceed by deriving some conditions on the incompatibility of two-outcome measurements based on the results of Prop. \ref{prop:inc-two-outcome-iff} that will help us prove that there exist incompatible measurements if and only if $\Ka$ is not a simplex.

The main idea is that we will construct two functions $f_1, f_2 \in A(\Ka)^+$ that reach both $0$ and $1$ on $\Ka$ and for the exposed faces
\begin{equation*}
F_i = \{ x \in \Ka : f_i(x) = 0 \},
\end{equation*}
$i \in \{1, 2\}$, it holds that $\conv( F_1 \cup F_2 ) \cap \intr(\Ka) \neq \emptyset$. Then by the Ineq. \eqref{eq:inc-two-outcome-ineq-1} and \eqref{eq:inc-two-outcome-ineq-2} we have that $p(x) = 0$ for every $x \in \conv(F_1 \cup F_2)$. Since $\conv( F_1 \cup F_2 ) \cap \intr(\Ka) \neq \emptyset$ and $p \geq 0$ we get $p = 0$. Then by Ineq. \eqref{eq:inc-two-outcome-ineq-3} we must have $f_1 + f_2 \leq 1$ if the measurements are compatible so we will show that we can construct functions $f_1, f_2$ with the mentioned properties such that $f_1(y) + f_2(y) > 1$ for some $y \in \Ka$ whenever $\Ka$ is not a simplex. 

The ideas presented above were inspired by an example of incompatible measurements on a square presented in \cite{BuschHeinosaariSchultzStevens-compatibility}.

\begin{prop} \label{prop:inc-two-outcome-maximal-face}
Let $x \in \Ka$ be an extreme point and let $F$ be a maximal face disjoint from $\{ x \}$, then there exist incompatible two-outcome measurements on $\Ka$ if $F$ does not contain all other extreme points of $\Ka$ except for $x$.
\end{prop}
\begin{proof}
For the definition of maximal face see Def. \ref{def:introd-faces-maximal} and remember that according to the definition $\Ka$ itself is not a maximal face. Note that closedness of $\Ka$ will play a role as it implies closedness of every face of $\Ka$ \cite[Corollary 18.1.1]{Rockafellar-convex}.

Assume that there is one maximal face $F$ disjoint from $x$, but $F$ does not contain all extreme points of $\Ka$ except for $x$, i.e. there is an extreme point $y \in \Ka$, such that $y \notin F$ and $y \neq x$. Since $F$, $\{x\}$, $\{y\}$ are closed sets and $\{y\}$ is disjoint from both $F$ and $\{x\}$, then there exists some open neighborhood $N_\varepsilon$ containing $y$, such that $x \notin N_\varepsilon$ and $F \cap N_\varepsilon = \emptyset$. There is an exposed point $z \in N_\varepsilon$ as the set of exposed points is dense in the set of extreme points of $\Ka$ \cite[Theorem 18.6]{Rockafellar-convex}. For the same reason we will consider $x$ an exposed point as well. Now let us construct positive affine function $f_1, f_x, f_z$ such that
\begin{align*}
F_1 &= \{ w \in \Ka : f_1(w) = 0 \}, \\
\{x\} &= \{ w \in \Ka : f_x(w) = 0 \}, \\
\{z\} &= \{ w \in \Ka : f_z(w) = 0 \}, \\
\end{align*}
and
\begin{equation*}
\max_{w \in \Ka} f_1(w) = \max_{w \in \Ka} f_x(w) = \max_{w \in \Ka} f_z(w) = 1.
\end{equation*}
The functions $f_1, f_x, f_z$ give rise to two-outcome measurements $m_1, m_x, m_z$ given as
\begin{align*}
m_1 &= f_1 \delta_{\omega_1} + (1 - f_1) \delta_{\omega_2}, \\
m_x &= f_x \delta_{\omega_1} + (1 - f_x) \delta_{\omega_2}, \\
m_z &= f_z \delta_{\omega_1} + (1 - f_z) \delta_{\omega_2}.
\end{align*}
Since we have
\begin{align*}
\conv( F_1 \cup \{x\}) \cap \intr(\Ka) &\neq \emptyset, \\
\conv( F_1 \cup \{z\}) \cap \intr(\Ka) &\neq \emptyset,
\end{align*}
we must have by Prop. \ref{prop:inc-two-outcome-iff}
\begin{align*}
&f_1 + f_x \leq 1,
&f_1 + f_z \leq 1,
\end{align*}
for the measurements $m_1, m_x$ and $m_1, m_z$ to be compatible. From $f_1 + f_x \leq 1$ we get $\{ w \in \Ka : f_1(w) = 1 \} = \{x\}$ and from $f_1 + f_z \leq 1$ we get $\{ w \in \Ka : f_1(w) = 1 \} = \{z\}$, which is a contradiction with $x \neq z$ implied by $x \notin N_\varepsilon$ and $z \in N_\varepsilon$.
\end{proof}

\begin{prop} \label{prop:inc-two-outcome-simplex}
Let $\Ka \subset \mathbb{R}^n$ be a compact convex set then there exist incompatible measurements on $\Ka$ whenever $\Ka$ is not a simplex.
\end{prop}
\begin{proof}
We will rely on the results of Prop. \ref{prop:inc-two-outcome-maximal-face}. Assume that $x \in \Ka$ is an extreme point that is affinely dependent on other extreme points, i.e. there are extreme points $\{y_1, \ldots, y_n\} \subset \Ka$ such that $x = \sum_{i=1}^n \alpha_i y_i$ with $\sum_{i=1}^n \alpha_i = 1$ and let $F$ denote the maximal face disjoint from $\{x\}$. Now let us construct a non-constant positive affine function $f \in A(\Ka)^+$ such that
\begin{equation*}
F = \{ z \in \Ka : f(z) = 0 \}.
\end{equation*}
Again the function $f$ exists as $F$ is an exposed face. Since $x = \sum_{i=1}^n \alpha_i y_i$, $\{y_1, \ldots, y_n\} \in F$ and $f$ is affine, we have
\begin{equation*}
f(x) = \sum_{i=1}^n \alpha_i f(y_i) = 0
\end{equation*}
and we must have $x \in F$, which is a contradiction. Hence the set of exposed points must be affinely independent, finite and $\Ka$ must be a simplex.
\end{proof}
It is an open question whether it can be in an easier fashion showed that the compactness and convexity of $\Ka$ together with compatibility of every two-outcome measurement implies the Riesz decomposition property \cite[pp. 84]{Alfsen-convSets} as it is known that it is equivalent to $\Ka$ being a simplex \cite[Corollary II.3.11]{Alfsen-convSets}. It is also known that in more general settings of effect algebras the result does not hold, i.e. there are effect algebras that are compatible but that do no satisfy Riesz Decomposition property, see \cite[Example 3.6]{Jenca-effAlg} for an example.

\subsection{Linear programming problem for compatibility of two-outcome measurements} \label{subsec:inc-linProg}
We will formulate the problem of compatibility of two two-outcome measurements as a problem of linear programming \cite{Barvinok-linProg} similar to the one obtained in \cite{WolfPerezgarciaFernandez-measIncomp}. We will start with the results of Prop. \ref{prop:inc-two-outcome-iff} and we will construct the linear programming problem from there.

Let $m_1, m_2$ be two-outcome measurements with sample space $\Omega = \{ \omega_1, \omega_2 \}$ given as
\begin{align*}
m_i &= f_i \delta_{\omega_1} + (1 - f_i) \delta_{\omega_2}
\end{align*}
for $i \in \{1, 2\}$ and let $\tau$ represent a coin-toss measurement given as
\begin{equation*}
\tau = \dfrac{1}{2} \left( \delta_{\omega_1} + \delta_{\omega_2} \right)
\end{equation*}
In the following calculations we will restrict ourselves only to this special coin-toss observable as it is sufficent to determine whether the measurements $m_1, m_2$ are compatible.

We want to know what is the highest possible $\lambda \in [\frac{1}{2}, 1]$, such that the measurements $\lambda m_1 + (1-\lambda) \tau, \lambda m_2 + (1-\lambda) \tau$ are compatible. In terms of Prop. \ref{prop:inc-two-outcome-iff} we want to know what is the highest value of $\lambda$ such that there exists $\tilde{p} \in A(\Ka)^+$ such that the conditions
\begin{align*}
\lambda f_1 + \dfrac{1-\lambda}{2} &\geq \tilde{p}, \\
\lambda f_2 + \dfrac{1-\lambda}{2} &\geq \tilde{p}, \\
1 + \tilde{p} &\geq \lambda ( f_1 + f_2) + (1-\lambda)
\end{align*}
are satisfied. Denoting $p = \frac{\tilde{p}}{\lambda}$ and $\mu = \frac{1-\lambda}{\lambda}$ we obtain
\begin{align}
\dfrac{\mu}{2} - p &\geq - f_1, \label{eq:inc-linProg-ineq1} \\
\dfrac{\mu}{2} - p &\geq - f_2, \label{eq:inc-linProg-ineq2} \\
p &\geq f_1 + f_2 - 1. \label{eq:inc-linProg-ineq3}
\end{align}
Now it is important to realize that maximizing $\lambda$ is equivalent to minimizing $\mu$. In the following we will introduce new partially ordered vector spaces and a linear map as the problem of linear programming will be formulated in their terms.

Let $x \in \mathbb{R} \times A(\Ka)$, then $x = ( \alpha, g)$ for $\alpha \in \mathbb{R}$ and $g \in A(\Ka)$. We introduce partial ordering on $\mathbb{R} \times A(\Ka)$ by the relation
\begin{equation*}
( \alpha, g) = x \geq 0 \quad \Leftrightarrow \quad \alpha \geq 0, \, g \in A(\Ka)^+.
\end{equation*}
The topological dual to $\mathbb{R} \times A(\Ka)$ is $\mathbb{R} \times A(\Ka)^*$, for $x = ( \alpha, g)$, $\tilde{c} \in \mathbb{R} \times A(\Ka)^*$, $\tilde{c} = (\beta, \psi)$, $\beta \in \mathbb{R}$, $\psi \in A(\Ka)^*$ we have
\begin{equation*}
\< \tilde{c}, x \> = \alpha \beta + \psi(g).
\end{equation*}
We will also use $A(\Ka) \times A(\Ka) \times A(\Ka)$ equipped with the following partial order: let $(g_1, g_2, g_3) \in A(\Ka) \times A(\Ka) \times A(\Ka)$, then $(g_1, g_2, g_3) \geq 0$ if and only if $g_i \geq 0$ for every $i \in \{ 1, 2, 3 \}$.

Let $T: \mathbb{R} \times A(\Ka) \to A(\Ka) \times A(\Ka) \times A(\Ka)$ be a linear map given as
\begin{equation*}
T (\alpha, g) = ( -g + \dfrac{\alpha}{2}, -g + \dfrac{\alpha}{2}, g),
\end{equation*}
where $\frac{\alpha}{2}$ stands for the constant function attaining the value $\frac{\alpha}{2}$. It is straightforward to see that $T$ is linear.

\begin{prop} \label{prop:inc-linProg-primal}
Let $c \in \mathbb{R} \times A(\Ka)^*$, $c = (1, 0)$, $F \in A(\Ka) \times A(\Ka) \times A(\Ka)$, $F = (-f_1, -f_2, f_1 + f_2 - 1)$ and $x \in \mathbb{R} \times A(\Ka)$, $x = (\mu, p)$ then
\begin{align*}
\inf &\< c , x \> \\
x &\geq 0 \\
Tx &\geq F
\end{align*}
is a primal linear programming problem. When the reached minimum is $0$ then the measurement $m_1, m_2$ are compatible. Moreover, there always exists primal feasible plan.
\end{prop}
\begin{proof}
The proof is straightforward. We have $\<c, x \> = \mu$ for the given $c$, $x \geq 0$ translates to $\mu \geq 0$ and $p \geq 0$. Note that $\mu \geq 0$ corresponds to $\lambda \leq 1$. $Tx \geq F$ is the same as
\begin{equation*}
( -p + \dfrac{\mu}{2}, -p + \dfrac{\mu}{2}, p ) \geq (-f_1, -f_2, f_1 + f_2 - 1)
\end{equation*}
which is in turn equivalent to conditions \eqref{eq:inc-linProg-ineq1} - \eqref{eq:inc-linProg-ineq3}.

Since $\mu = \frac{1-\lambda}{\lambda}$ then $\mu = 0$ implies $\lambda = 1$. There always exists a primal feasible plan as we know that for $\lambda = \frac{1}{2}$ the measurements are always compatible, see Prop. \ref{prop:inc-minimal-deg}.
\end{proof}

Now that we have the primal problem we will find the dual problem to obtain another condition on the compatibility of measurements $m_1, m_2$.
\begin{prop} \label{prop:inc-linProg-dual}
The dual problem to problem introduced in Prop. \ref{prop:inc-linProg-primal} is given as
\begin{align*}
\sup &\< F , l \> \\
T^* l &\leq c \\
l &\geq 0
\end{align*}
where $l \in A(\Ka)^* \times A(\Ka)^* \times A(\Ka)^*$ and $T^*$ is given by the relation $\< \tilde{l}, T \tilde{x} \> = \< T^* \tilde{l}, \tilde{x} \>$ for every $\tilde{l} \in A(\Ka)^* \times A(\Ka)^* \times A(\Ka)^*$ and $\tilde{x} \in \mathbb{R} \times A(\Ka)$, i.e. $T^*: A(\Ka)^* \times A(\Ka)^* \times A(\Ka)^* \to \mathbb{R} \times A(\Ka)^*$, such that for $(\psi_1, \psi_2, \psi_3) \in A(\Ka)^* \times A(\Ka)^* \times A(\Ka)^*$ we have
\begin{equation*}
T^* ( \psi_1, \psi_2, \psi_3) = \left( \dfrac{1}{2}(\psi_1 + \psi_2)(1), - \psi_1 - \psi_2 + \psi_3 \right)
\end{equation*}
where $1$ stands for the constant function on $\Ka$ and $\psi_i(1)$ is the value of functional $\psi_1$ on this function, that is for some $z_{11}, z_{12} \in \Ka$ and $a_1, a_2 \in \mathbb{R}$, $a_1 \geq 0$, $a_2 \geq 0$ we have $\psi_1 = a_1 \phi_{z_1} - a_2 \phi_{z_2}$ and $\psi_1 (1) = a_1 - a_2$.
\end{prop}
\begin{proof}
The only thing we need to do is to find $T^*$, the rest follows from the relation between primal and dual problems \cite[pp. 163]{Barvinok-linProg}.

From the relation $\< \tilde{l}, T \tilde{x} \> = \< T^* \tilde{l}, \tilde{x} \>$ for $\tilde{l} = (\psi_1, \psi_2, \psi_3) \in A(\Ka)^* \times A(\Ka)^* \times A(\Ka)^*$ and $\tilde{x} = (\alpha, g) \in \mathbb{R} \times A(\Ka)$ we get
\begin{align*}
\< \tilde{l}, T \tilde{x} \> &= \left\langle (\psi_1, \psi_2, \psi_3), \left( -g + \dfrac{\alpha}{2}, -g + \dfrac{\alpha}{2}, g \right) \right\rangle  \\
&= \dfrac{\alpha}{2} (\psi_1 + \psi_2)(1) + (-\psi_1 - \psi_2 + \psi_3)(g) \\
&= \left\langle \left( \dfrac{1}{2}(\psi_1 + \psi_2)(1), - \psi_1 - \psi_2 + \psi_3 \right), ( \alpha, g) \right\rangle \\
&= \< T^* \tilde{l}, \tilde{x} \>.
\end{align*}
\end{proof}

\begin{prop} \label{prop:inc-linProg-zeroGap}
The duality gap between the primal problem given by Prop. \ref{prop:inc-linProg-primal} and the dual problem given by Prop. \ref{prop:inc-linProg-dual} is zero.
\end{prop}
\begin{proof}
The duality gap is zero if there is a primal feasible plan and the cone
\begin{align*}
Q &= \{ (T \tilde{x}, \<c, \tilde{x} \>) : \tilde{x} \in \mathbb{R} \times A(\Ka), \tilde{x} \geq 0 \}, \\
Q &\subset A(\Ka) \times A(\Ka) \times A(\Ka) \times \mathbb{R},
\end{align*}
where $c = (1, 0)$ as in Prop. \ref{prop:inc-linProg-primal}, is closed \cite[Theorem 7.2]{Barvinok-linProg}. To show that $Q$ is closed we will use the fact that if $V, W$ are topological vector spaces, $Q_V \subset V$ is a cone with compact convex base and $T_V : V \to W$ is a continuous linear transformation, such that $\ker(T_V) \cap Q_V = \{ 0 \}$, then the cone $T_V(Q_V)$ is closed \cite[Lemma 7.3]{Barvinok-linProg}.

Because the cone $A(K)^+$ is generating there exists a base of positive functions $h_1, \ldots h_n$ such that for every $\tilde{h} \in A(K)^+$ we have $\tilde{h} = \sum_{i=1}^n \lambda_i h_i$ for $\lambda_i \geq 0$. We introduce the $L^1$ norm on $A(K)$: for $h' \in A(K)$, $h' = \sum_{i=1}^n \nu_i h_i$ we have $\Vert h' \Vert_{L1} = \sum_{i=1}^n \vert \nu_i \vert$. Note that this norm is an affine function on $A(K)^+$.

We can introduce a norm on $\mathbb{R} \times A(\Ka)$ as follows: let $\tilde{x} = (\alpha, g) \in \mathbb{R} \times A(\Ka)$, then
\begin{equation*}
\Vert \tilde{x} \Vert_{\mathbb{R} \times A(\Ka)} = |\alpha| + \Vert g \Vert_{L1}.
\end{equation*}

The base of the positive cone in $\mathbb{R} \times A(\Ka)$ is the set
\begin{equation*}
\mathcal{\Ka} = \{ \tilde{x} \in \mathbb{R} \times A(\Ka) : \Vert \tilde{x} \Vert_{\mathbb{R} \times A(\Ka)} = 1 \}.
\end{equation*}
$\mathcal{K}$ is compact and convex, because the norm $\Vert \cdot \Vert_{\mathbb{R} \times A(\Ka)}$ is continuous and for $\alpha \geq 0$ and $g \in A(K)^+$ it is affine.

The map $T': \mathbb{R} \times A(\Ka) \to A(\Ka) \times A(\Ka) \times A(\Ka) \times \mathbb{R}$ given as
\begin{equation*}
T' \tilde{x} = ( T\tilde{x}, \<c, \tilde{x} \> )
\end{equation*}
is linear and continuous. If for $(\alpha, g) = \tilde{x} \in \mathbb{R} \times A(\Ka)$ hold that $T' \tilde{x} = 0$ then we have to have $\tilde{x} = (0, 0)$ as $\<c, \tilde{x} \> = 0$ implies $\alpha = 0$ and $T \tilde{x} = (0, 0, 0)$ implies $g = 0$. In conclusion we have $\ker(T') = \{ (0, 0) \}$.

This shows that the cone $Q$ is closed and since we have already showed in Prop. \ref{prop:inc-linProg-primal} that primal feasible plan exists, the duality gap is zero.
\end{proof}

We will proceed with rewriting the dual problem from Prop. \ref{prop:inc-linProg-dual} into a more usable form to obtain necessary and sufficient condition for two two-outcome measurements to be incompatible. We will start from the dual problem stated in Prop. \ref{prop:inc-linProg-dual}. Since $l \in A(\Ka)^* \times A(\Ka)^* \times A(\Ka)^*$ and $l \geq 0$ we must have some $z_1, z_2, z_3 \in \Ka$ and $a_1, a_2, a_3 \in \mathbb{R}$, $a_i \geq 0$, $i \in \{1, 2, 3\}$, such that $l = (a_1 \phi_{z_1}, a_2 \phi_{z_2}, a_3 \phi_{z_3})$ in the formalism of Subsec. \ref{subsec:structure}. From $T^* l \leq c$ we obtain the conditions
\begin{align}
\dfrac{1}{2} ( a_1 + a_2 ) &\leq 1, \label{eq:inc-linProg-dual-ineq1} \\
a_3 \phi_{z_3} &\leq a_1 \phi_{z_1} + a_2 \phi_{z_2}. \label{eq:inc-linProg-dual-ineq2}
\end{align}
Moreover we have
\begin{equation*}
\< F, l \> = -a_1 f_1(z_1) - a_2 f_2(z_2) + a_3 ( f_1 (z_3) + f_2(z_3) - 1 ).
\end{equation*}
Thus we have proved:
\begin{prop}
The two-outcome measurements $m_1, m_2$ corresponding to the functions $f_1, f_2$ are incompatible if and only if there exists positive numbers $a_1, a_2, a_3 \in \mathbb{R}$ and $z_1, z_2, z_3 \in \Ka$ such that Ineq. \eqref{eq:inc-linProg-dual-ineq1}, \eqref{eq:inc-linProg-dual-ineq2} are satisfied and $\< F, l \> > 0$.
\end{prop}
If one would wish to have $\< F, l \> > 0$ it would be a possible first idea to have $f_1(z_1) = f_2(z_2) = 0$ as then only $f_1 (z_3) + f_2(z_3) > 1$ would be required. In this case it would be easy to satisfy the Ineq. \eqref{eq:inc-linProg-dual-ineq2} by suitable choice of $a_3$ whenever $\conv( \{z_1, z_2 \} ) \cap \intr(\Ka) \neq \emptyset$ as then for some $\nu \in [0,1]$ we would have $\nu z_1 + (1-\nu) z_2 \in \intr(\Ka)$ and $\phi_{\nu z_1 + (1-\nu) z_2}$ would be an order unit in $A(\Ka)^*$. Matter of fact, this is exactly the idea we used to prove Prop. \ref{prop:inc-two-outcome-simplex}.

It is worth mentioning that by similar methods of semidefinite programming it was shown that in case of measurements on states that the value of $\< F, l \>$ corresponds to maximal violation of CHSH Bell inequality \cite{WolfPerezgarciaFernandez-measIncomp}.

\section{Conclusions}
Incompatibility of measurements is one of the key aspects of quantum theories and as our results have shown, in finite dimensional cases it only differentiates classical probabilistic theories from general probabilistic theories. The quest for finding some essentially quantum restriction on probabilistic theories also considered in \cite{BarnumHowardBarretLeifer-noBroadcast} is not over as such restriction would probably help us understand quantum theories better and deeper.

It is of course an open question whether such aspect of quantum theories that would differentiate it from other non-classical probabilistic theories exists. It is also an open question whether our results hold also in the infinite dimensional case. Possible approach to generalize our results to infinite dimensional case would be to prove it using Riesz decomposition property and to observe whether the proof may be generalized for infinite-dimensional state space.

\section*{Acknowledgments}
The author is thankful to Anna Jen\v{c}ov\'{a} for helpful and stimulating conversations and to M\'{a}rio Ziman, Teiko Heinosaari and Takayuki Miyadera for sharing their research notes. This research was supported by grant VEGA 2/0069/16.

\bibliography{citations}

\begin{thebibliography}{16}%
\makeatletter
\providecommand \@ifxundefined [1]{%
 \@ifx{#1\undefined}
}%
\providecommand \@ifnum [1]{%
 \ifnum #1\expandafter \@firstoftwo
 \else \expandafter \@secondoftwo
 \fi
}%
\providecommand \@ifx [1]{%
 \ifx #1\expandafter \@firstoftwo
 \else \expandafter \@secondoftwo
 \fi
}%
\providecommand \natexlab [1]{#1}%
\providecommand \enquote  [1]{``#1''}%
\providecommand \bibnamefont  [1]{#1}%
\providecommand \bibfnamefont [1]{#1}%
\providecommand \citenamefont [1]{#1}%
\providecommand \href@noop [0]{\@secondoftwo}%
\providecommand \href [0]{\begingroup \@sanitize@url \@href}%
\providecommand \@href[1]{\@@startlink{#1}\@@href}%
\providecommand \@@href[1]{\endgroup#1\@@endlink}%
\providecommand \@sanitize@url [0]{\catcode `\\12\catcode `\$12\catcode
  `\&12\catcode `\#12\catcode `\^12\catcode `\_12\catcode `\%12\relax}%
\providecommand \@@startlink[1]{}%
\providecommand \@@endlink[0]{}%
\providecommand \url  [0]{\begingroup\@sanitize@url \@url }%
\providecommand \@url [1]{\endgroup\@href {#1}{\urlprefix }}%
\providecommand \urlprefix  [0]{URL }%
\providecommand \Eprint [0]{\href }%
\providecommand \doibase [0]{http://dx.doi.org/}%
\providecommand \selectlanguage [0]{\@gobble}%
\providecommand \bibinfo  [0]{\@secondoftwo}%
\providecommand \bibfield  [0]{\@secondoftwo}%
\providecommand \translation [1]{[#1]}%
\providecommand \BibitemOpen [0]{}%
\providecommand \bibitemStop [0]{}%
\providecommand \bibitemNoStop [0]{.\EOS\space}%
\providecommand \EOS [0]{\spacefactor3000\relax}%
\providecommand \BibitemShut  [1]{\csname bibitem#1\endcsname}%
\let\auto@bib@innerbib\@empty
\bibitem [{\citenamefont {Heinosaari}\ \emph
  {et~al.}(2015{\natexlab{a}})\citenamefont {Heinosaari}, \citenamefont
  {Miyadera},\ and\ \citenamefont
  {Ziman}}]{HeinosaariMiyaderaZiman-compatibility}%
  \BibitemOpen
  \bibfield  {author} {\bibinfo {author} {\bibfnamefont {T.}~\bibnamefont
  {Heinosaari}}, \bibinfo {author} {\bibfnamefont {T.}~\bibnamefont
  {Miyadera}}, \ and\ \bibinfo {author} {\bibfnamefont {M.~M.}\ \bibnamefont
  {Ziman}},\ }\href@noop {} {\bibfield  {journal} {\bibinfo  {journal} {J.
  Phys. A Math. Theor.}\ }\textbf {\bibinfo {volume} {49}},\ \bibinfo {pages}
  {123001} (\bibinfo {year} {2015}{\natexlab{a}})},\ \Eprint
  {http://arxiv.org/abs/1511.07548} {arXiv:1511.07548} \BibitemShut {NoStop}%
\bibitem [{\citenamefont {Wolf}\ \emph {et~al.}(2009)\citenamefont {Wolf},
  \citenamefont {Perez-Garcia},\ and\ \citenamefont
  {Fernandez}}]{WolfPerezgarciaFernandez-measIncomp}%
  \BibitemOpen
  \bibfield  {author} {\bibinfo {author} {\bibfnamefont {M.~M.}\ \bibnamefont
  {Wolf}}, \bibinfo {author} {\bibfnamefont {D.}~\bibnamefont {Perez-Garcia}},
  \ and\ \bibinfo {author} {\bibfnamefont {C.}~\bibnamefont {Fernandez}},\
  }\href@noop {} {\bibfield  {journal} {\bibinfo  {journal} {Phys. Rev. Lett.}\
  }\textbf {\bibinfo {volume} {103}},\ \bibinfo {pages} {1} (\bibinfo {year}
  {2009})},\ \Eprint {http://arxiv.org/abs/0905.2998} {arXiv:0905.2998}
  \BibitemShut {NoStop}%
\bibitem [{\citenamefont {Barnum}\ \emph {et~al.}(2007)\citenamefont {Barnum},
  \citenamefont {Barrett}, \citenamefont {Leifer},\ and\ \citenamefont
  {Wilce}}]{BarnumHowardBarretLeifer-noBroadcast}%
  \BibitemOpen
  \bibfield  {author} {\bibinfo {author} {\bibfnamefont {H.}~\bibnamefont
  {Barnum}}, \bibinfo {author} {\bibfnamefont {J.}~\bibnamefont {Barrett}},
  \bibinfo {author} {\bibfnamefont {M.}~\bibnamefont {Leifer}}, \ and\ \bibinfo
  {author} {\bibfnamefont {A.}~\bibnamefont {Wilce}},\ }\href@noop {}
  {\bibfield  {journal} {\bibinfo  {journal} {Phys. Rev. Lett.}\ }\textbf
  {\bibinfo {volume} {99}},\ \bibinfo {pages} {1} (\bibinfo {year} {2007})},\
  \Eprint {http://arxiv.org/abs/0707.0620} {arXiv:0707.0620} \BibitemShut
  {NoStop}%
\bibitem [{\citenamefont {Busch}\ \emph {et~al.}(2013)\citenamefont {Busch},
  \citenamefont {Heinosaari}, \citenamefont {Schultz},\ and\ \citenamefont
  {Stevens}}]{BuschHeinosaariSchultzStevens-compatibility}%
  \BibitemOpen
  \bibfield  {author} {\bibinfo {author} {\bibfnamefont {P.}~\bibnamefont
  {Busch}}, \bibinfo {author} {\bibfnamefont {T.}~\bibnamefont {Heinosaari}},
  \bibinfo {author} {\bibfnamefont {J.}~\bibnamefont {Schultz}}, \ and\
  \bibinfo {author} {\bibfnamefont {N.}~\bibnamefont {Stevens}},\ }\href@noop
  {} {\bibfield  {journal} {\bibinfo  {journal} {EPL (Europhysics Lett.)}\
  }\textbf {\bibinfo {volume} {103}},\ \bibinfo {pages} {10002} (\bibinfo
  {year} {2013})},\ \Eprint {http://arxiv.org/abs/1210.4142} {arXiv:1210.4142}
  \BibitemShut {NoStop}%
\bibitem [{\citenamefont {Banik}(2015)}]{Banik-steering}%
  \BibitemOpen
  \bibfield  {author} {\bibinfo {author} {\bibfnamefont {M.}~\bibnamefont
  {Banik}},\ }\href@noop {} {\bibfield  {journal} {\bibinfo  {journal} {J.
  Math. Phys.}\ }\textbf {\bibinfo {volume} {56}},\ \bibinfo {pages} {1}
  (\bibinfo {year} {2015})},\ \Eprint {http://arxiv.org/abs/1502.05779}
  {arXiv:1502.05779} \BibitemShut {NoStop}%
\bibitem [{\citenamefont {Uola}\ \emph {et~al.}(2014)\citenamefont {Uola},
  \citenamefont {Moroder},\ and\ \citenamefont
  {G{\"{u}}hne}}]{UolaMoroder-steering}%
  \BibitemOpen
  \bibfield  {author} {\bibinfo {author} {\bibfnamefont {R.}~\bibnamefont
  {Uola}}, \bibinfo {author} {\bibfnamefont {T.}~\bibnamefont {Moroder}}, \
  and\ \bibinfo {author} {\bibfnamefont {O.}~\bibnamefont {G{\"{u}}hne}},\
  }\href@noop {} {\bibfield  {journal} {\bibinfo  {journal} {Phys. Rev. Lett.}\
  }\textbf {\bibinfo {volume} {113}},\ \bibinfo {pages} {1} (\bibinfo {year}
  {2014})},\ \Eprint {http://arxiv.org/abs/1407.2224} {arXiv:1407.2224}
  \BibitemShut {NoStop}%
\bibitem [{\citenamefont {Sedl{\'{a}}k}\ \emph {et~al.}(2016)\citenamefont
  {Sedl{\'{a}}k}, \citenamefont {Reitzner}, \citenamefont {Chiribella},\ and\
  \citenamefont {Ziman}}]{SedlakReitznerChiribellaZiman-compatibility}%
  \BibitemOpen
  \bibfield  {author} {\bibinfo {author} {\bibfnamefont {M.}~\bibnamefont
  {Sedl{\'{a}}k}}, \bibinfo {author} {\bibfnamefont {D.}~\bibnamefont
  {Reitzner}}, \bibinfo {author} {\bibfnamefont {G.}~\bibnamefont
  {Chiribella}}, \ and\ \bibinfo {author} {\bibfnamefont {M.}~\bibnamefont
  {Ziman}},\ }\href@noop {} {\bibfield  {journal} {\bibinfo  {journal} {Phys.
  Rev. A}\ }\textbf {\bibinfo {volume} {93}},\ \bibinfo {pages} {1} (\bibinfo
  {year} {2016})}\BibitemShut {NoStop}%
\bibitem [{\citenamefont {Heinosaari}\ \emph
  {et~al.}(2015{\natexlab{b}})\citenamefont {Heinosaari}, \citenamefont
  {Kiukas},\ and\ \citenamefont
  {Reitzner}}]{HeinosaariKiukasReitzner-robustness}%
  \BibitemOpen
  \bibfield  {author} {\bibinfo {author} {\bibfnamefont {T.}~\bibnamefont
  {Heinosaari}}, \bibinfo {author} {\bibfnamefont {J.}~\bibnamefont {Kiukas}},
  \ and\ \bibinfo {author} {\bibfnamefont {D.}~\bibnamefont {Reitzner}},\
  }\href@noop {} {\bibfield  {journal} {\bibinfo  {journal} {Phys. Rev. A - At.
  Mol. Opt. Phys.}\ }\textbf {\bibinfo {volume} {92}},\ \bibinfo {pages} {1}
  (\bibinfo {year} {2015}{\natexlab{b}})},\ \Eprint
  {http://arxiv.org/abs/1501.04554} {arXiv:1501.04554} \BibitemShut {NoStop}%
\bibitem [{\citenamefont {Rockafellar}(1997)}]{Rockafellar-convex}%
  \BibitemOpen
  \bibfield  {author} {\bibinfo {author} {\bibfnamefont {R.~T.}\ \bibnamefont
  {Rockafellar}},\ }\href@noop {} {\emph {\bibinfo {title} {Convex
  Analysis}}},\ Princeton landmarks in mathematics and physics\ (\bibinfo
  {publisher} {Princeton University Press},\ \bibinfo {year}
  {1997})\BibitemShut {NoStop}%
\bibitem [{\citenamefont {Asimow}\ and\ \citenamefont
  {Ellis}(1980)}]{AsimowEllis}%
  \BibitemOpen
  \bibfield  {author} {\bibinfo {author} {\bibfnamefont {L.}~\bibnamefont
  {Asimow}}\ and\ \bibinfo {author} {\bibfnamefont {A.~J.}\ \bibnamefont
  {Ellis}},\ }\href@noop {} {\emph {\bibinfo {title} {Convexity theory and its
  applications in functional analysis}}},\ L.M.S. monographs\ (\bibinfo
  {publisher} {Academic Press},\ \bibinfo {year} {1980})\BibitemShut {NoStop}%
\bibitem [{\citenamefont {Heinosaari}\ and\ \citenamefont
  {Ziman}(2012)}]{HeinosaariZiman-MLQT}%
  \BibitemOpen
  \bibfield  {author} {\bibinfo {author} {\bibfnamefont {T.}~\bibnamefont
  {Heinosaari}}\ and\ \bibinfo {author} {\bibfnamefont {M.}~\bibnamefont
  {Ziman}},\ }\href@noop {} {\emph {\bibinfo {title} {The Mathematical Language
  of Quantum Theory. From Uncertainty to Entanglement}}}\ (\bibinfo
  {publisher} {Cambridge University Press},\ \bibinfo {year}
  {2012})\BibitemShut {NoStop}%
\bibitem [{\citenamefont {Holevo}(2011)}]{Holevo-QT}%
  \BibitemOpen
  \bibfield  {author} {\bibinfo {author} {\bibfnamefont {A.~S.}\ \bibnamefont
  {Holevo}},\ }\href@noop {} {\emph {\bibinfo {title} {Probabilistic and
  Statistical Aspects of Quantum Theory}}},\ Publications of the Scuola Normale
  Superiore\ (\bibinfo  {publisher} {Scuola Normale Superiore},\ \bibinfo
  {year} {2011})\BibitemShut {NoStop}%
\bibitem [{\citenamefont {Heinosaari}\ and\ \citenamefont
  {Miyadera}()}]{HeinosaariMiyadera-compOfChan}%
  \BibitemOpen
  \bibfield  {author} {\bibinfo {author} {\bibfnamefont {T.}~\bibnamefont
  {Heinosaari}}\ and\ \bibinfo {author} {\bibfnamefont {T.}~\bibnamefont
  {Miyadera}},\ }\href {http://arxiv.org/abs/1608.01794} {\ }\Eprint
  {http://arxiv.org/abs/1608.01794} {arXiv:1608.01794} \BibitemShut {NoStop}%
\bibitem [{\citenamefont {Alfsen}(1971)}]{Alfsen-convSets}%
  \BibitemOpen
  \bibfield  {author} {\bibinfo {author} {\bibfnamefont {E.~M.}\ \bibnamefont
  {Alfsen}},\ }\href@noop {} {\emph {\bibinfo {title} {Compact convex sets and
  boundary integrals}}},\ Ergebnisse der Mathematik und ihrer Grenzgebiete\
  (\bibinfo  {publisher} {Springer-Verlag},\ \bibinfo {year}
  {1971})\BibitemShut {NoStop}%
\bibitem [{\citenamefont {Jen{\v{c}}a}(2015)}]{Jenca-effAlg}%
  \BibitemOpen
  \bibfield  {author} {\bibinfo {author} {\bibfnamefont {G.}~\bibnamefont
  {Jen{\v{c}}a}},\ }\href@noop {} {\bibfield  {journal} {\bibinfo  {journal}
  {Bull. Aust. Math. Soc.}\ }\textbf {\bibinfo {volume} {64}},\ \bibinfo
  {pages} {81} (\bibinfo {year} {2015})},\ \Eprint
  {http://arxiv.org/abs/1504.00354} {arXiv:1504.00354} \BibitemShut {NoStop}%
\bibitem [{\citenamefont {Barvinok}(2002)}]{Barvinok-linProg}%
  \BibitemOpen
  \bibfield  {author} {\bibinfo {author} {\bibfnamefont {A.}~\bibnamefont
  {Barvinok}},\ }\href@noop {} {\emph {\bibinfo {title} {A Course in
  Convexity}}},\ Graduate studies in mathematics\ (\bibinfo  {publisher}
  {American Mathematical Society},\ \bibinfo {year} {2002})\BibitemShut
  {NoStop}%
\end{thebibliography}%

\end{document}